\documentclass[draftcls,11pt,onecolumn]{IEEEtran}

\ifCLASSINFOpdf
\else
\fi

\usepackage{times,amsmath,color,amssymb,epsfig,cite,enumerate,setspace}
\usepackage{subfigure,multirow,bm,stfloats}
\usepackage[para]{threeparttable}

\usepackage{tabularx}
\usepackage{array}
\usepackage{booktabs}
\usepackage{graphicx}
\usepackage{graphics}
\usepackage{pbox}
\usepackage{dsfont}
\usepackage{psfrag}

\newtheorem{Remark}{\it Remark}[section]

\newtheorem{Proposition}{\it Proposition}[section]
\newtheorem{Lemma}{\it Lemma}[section]

\begin{document}
\begin{spacing}{1.5}
%
\title{Opportunistic Multi-Channel Access in Heterogeneous 5G Network with Renewable Energy Supplies}
\author{Hang~Li, 
        Chuan~Huang, 
        Fuad E. Alsaadi, 
        Abdullah M. Dobaie,
        and~Shuguang~Cui 
\thanks{Part of this work appeared in the Proceedings of the 6th International Symposium on Communications Control, and Signal Processing (ISCCSP), Athens, Greece, May 21- 23, 2014.}
\thanks{H.~Li and S.~Cui are with the Department of Electrical and Computer Engineering, Texas A\&M University, College Station, Texas, 77843 USA (e-mail: david\_lihang@tamu.edu; cui@ece.tamu.edu).}
\thanks{C.~Huang is with the National Key Laboratory of Science and Technology on Communications, University of Electronic Science and Technology of China, Chengdu, Sichuan 610051 China (e-mail: huangch@uestc.edu.cn).}
\thanks{F. E. Alsaadi and A. M. Dobaie are with the Department of Electrical and Computer Engineering, King Abdulaziz University, Jeddah, 22254 Saudi Arabia (e-mail: fuad\_alsaadi@yahoo.com; adobaie@kau.edu.sa).}}
\maketitle
\begin{abstract}
A heterogeneous system, where small networks (e.g., small cell or WiFi) boost the system throughput under the umbrella of a large network (e.g., large cell), is a promising architecture for the 5G wireless communication networks, where green and sustainable communication is also a key aspect. Renewable energy based communication via energy harvesting (EH) devices is one of such green technology candidates. In this paper, we study an uplink transmission scenario under a heterogeneous network hierarchy, where each mobile user (MU) is powered by a sustainable energy supply, capable of both deterministic access to the large network via one private channel, and dynamic access to a small network with certain probability via one common channel shared by multiple MUs. Considering a general EH model, i.e., energy arrivals are time-correlated, we study an opportunistic transmission scheme and aim to maximize the average throughput for each MU, which jointly exploits the statistics and current states of the private channel, common channel, battery level, and EH rate. Applying a simple yet efficient ``save-then-transmit'' scheme, the throughput maximization problem is cast as a ``rate-of-return'' optimal stopping problem. The optimal stopping rule is proved to has a time-dependent threshold-based structure for the case with general Markovian system dynamics, and degrades to a pure threshold policy for the case with independent and identically distributed system dynamics. As performance benchmarks, the optimal power allocation scheme with conventional power supplies is also examined. Finally, numerical results are presented, and a new concept of ``EH diversity'' is discussed.
\end{abstract}
\begin{IEEEkeywords}
Heterogeneous networks, small cell, energy harvesting, opportunistic transmission, optimal stopping.
\end{IEEEkeywords}

\section{Introduction}
\subsection{Motivations}
Heterogeneous networks (HetNets), where small networks (e.g., small cell or WiFi) composed of low-power access points (APs) are placed under the coverage of a large network (e.g., large cell), are evolving into a new type of network deployment that could enhance the overall system throughput with reasonable cost and power consumption \cite{JGAndrews2013,AG}. Standardization bodies, such as ETSI and 3GPP, have paid much attention to this shifting of network paradigm and have made HetNets part of the current and future cellular standards. Now, commercial small cell deployments could already be found globally, operated by various cellular carriers \cite{smallcell,JGAndrews2012}.

In a traditional cellular network, a mobile user (MU) is usually assigned a dedicated private channel to access the base station (BS), while this link may experience bad channel conditions due to the possible severe path loss and shadowing between the MU and the BS. In such cases, however, the desired quality-of-service (QoS) could still be satisfied by allowing the MU to access a nearby AP in an underlying small network via a common channel, if the corresponding channel condition is relatively good. Essentially, the MU in the above HetNet could deploy a multi-channel access scheme: The messages from MU could be directly delivered to the cellular BS, or if available, jointly via a nearby low-power AP \cite{DLP}. It is worth noting that the small network could be operated over a band orthogonal to the large network: e.g., WiFi uses the unlicensed band \cite{MMAW} and femtocells could be allocated with different bands from the large network via orthogonal frequency division multiple access (OFDMA) or time division multiple access (TDMA) \cite{JGAndrews2012,XiaAndrews2010}. If needed, the small network can also share the same bands with the large network. For either case, there are two modes of access control for small networks: restricted access, i.e., only pre-registered MUs could access the corresponding AP \cite{JGAndrews2012,DLP}; and open access, i.e., any local MUs in the small network could gain the access. In practice, the MU may fail to establish a dedicated link to the small network due to congestion over the limited spectrum resources, which introduces another type of access randomness beyond channel variation in the conventional cellular system.

Another significant advantage enabled by the aforementioned HetNet is that the MU could potentially enjoy a longer lifetime since its power consumption may be reduced by the help of communicating with the nearby local AP. However, since the lifetime of an MU is still limited by the stored energy in the batteries \cite{HoZhang2012}, the MU should seek an ``active'' way to recharge itself, especially in a green fashion. Such renewable energy powered nodes, which can efficiently convert certain environment energy (e.g., those from solar, wind, and vibration) into electric energy \cite{SS}, will play critical roles in the next generation or 5G wireless system, which is designed to be environment friendly and to support diversified applications such as machine-to-machine communications and Internet of things (IoT). In this way, the MU could prolong the battery life almost infinitely, and fulfil the increasing demands of green operations in 5G \cite{Wang2012}. Compared with the conventional power supply, such a renewable energy supply raises a new transmission design constraint: The consumed energy up to any time should be bounded by the harvested energy until this point, which is named as the EH constraint \cite{HoZhang2012}.

\begin{figure}
  \centering
  \includegraphics[width=3.3in]{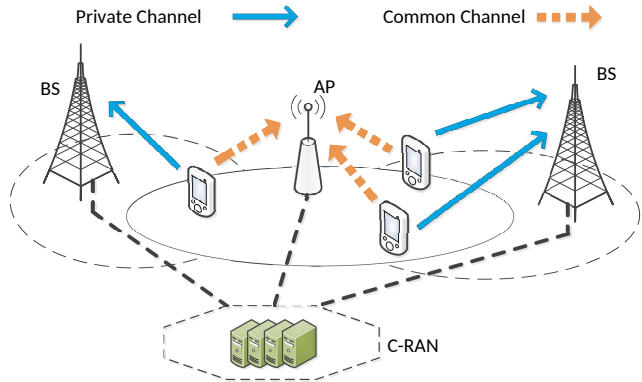}
  \caption{The uplink HetNet with multi-channel access with C-RAN platform, where each MU is powered by energy harvesters, and has accesses to the BS and AP via private and common channels, respectively.}
  \label{chamodel}
\end{figure}

In this paper, we study a simple uplink HetNet scenario depicted in Fig. \ref{chamodel}, where each EH-based MU has an individual link, namely a {\it private channel}, to the large network BS for deterministic access. Moreover, a local AP of a small network offers a {\it common channel}, which is randomly shared by all nearby MUs. Here we consider a scenario that each MU could access the common channel with a certain probability at each time slot. Thus, based on this multi-channel access setup, the MU could fulfil a transmission by using the harvested energy via either its private channel solely or via both the private and common channels simultaneously. Joint information processing is done with low latency by a cloud-based radio access network (C-RAN) platform, which is a popular platform candidate for 5G \cite{CRAN,FluidNet,Chilini}.

On the MU side, there are two types of state information that could be causally known before the transmission: the channel state information (CSI) of the links to the large network and the small network (if the AP was successfully accessed by the MU); and the energy state information (ESI), i.e., the EH rate (the harvested energy per unit time) and the battery state at the MU. Therefore, the MU could decide when to start a transmission with both CSI and ESI at hand. Obviously, a longer time to harvest energy while probing the system may result in a higher transmission power, and create a higher likelihood to secure the common channel; however, it may reduce the average effective transmission time. Thus, this leaves us an interesting tradeoff to optimize: energy saving time vs. data transmission time. In addition, we consider a ``save-then-transmit'' scheme such that each transmission would consume all the harvested energy at the MU. This suboptimal power utilization scheme is able to deploy a large instantaneous transmit power such that the short-term transmission rate is maximized, and is more tractable for analysis as well.

\subsection{Contributions}
First, we propose an opportunistic transmission scheme for the multi-channel HetNet uplink powered by sustainable energy supplies, which enhances the average throughput for each user by jointly exploiting the stochastic CSI and ESI. More precisely, the throughput maximization is cast as a ``rate-of-return'' optimal stopping problem. With Markovian private channel and EH models, the optimal stopping rule is proved to exist and have a state-dependent threshold-based structure under both finite and infinite battery capacity assumptions. The optimal throughput is proved to be strictly increasing over the probability that the common channel is secured.

Second, we study the case when the private channel gains and the EH rates are respectively independent and identically distributed (i.i.d.) across different communication blocks. The corresponding optimal stopping rule is proved to be a pure-threshold policy, i.e., the threshold does not change over time, which could be found via a one-dimension search. With such a fixed threshold, the mean saving time is proved to be decreasing polynomially over the probability that the common channel is secured. We also show via simulations that the randomness of EH rates, leading to the so-called ``EH diversity'', influences the throughput performance and could be exploited by our proposed pure-threshold policy: Specifically, we find that the more dynamically the EH rate varies, the higher the average throughput that the MU could achieve.

Finally, we quantify the performance of the case with conventional power supplies as the benchmark, showing that the corresponding optimal power allocation has a ``water-filling'' structure, where the water level is jointly determined by the statistics of the private and common channels, and the probability that the common channel is secured.

\subsection{Related Works}
Most of existing works related to the uplink of heterogeneous cellular networks assume certain deterministic access control of the underlying small networks \cite{JGAndrews2012,XiaAndrews2010,VCJA,MBennis}. From the views of both the femtocell owner and the overall network operator, authors in \cite{XiaAndrews2010} evaluated the femtocell performance with open and restricted accesses. It was shown that with nonorthogonal (in terms of frequency or time) multiple access for mobile users, open access benefits both the femtocell owner and the network operator; with orthogonal multiple access, the femtocell access control strategy (open or restricted) is closely dependent on the user density. In \cite{VCJA}, by adopting open access, the outage behaviors of both femtocell and large cell users were analyzed via stochastic geometry to model the locations of both the femtocell APs and the cellular users. The authors also presented several interference avoidance methods to enhance the per-user capacity. In \cite{MBennis}, each large cell user was assigned one direct link to the BS, and one relay link to the femtocell AP. Playing a non-cooperative game against the others, each user could seek its preferred open-access femtocell and split the rates between the BS and the AP to maximize its own utility. In contrast to these existing works, here we consider users with random, not deterministic, access to the local AP, which is more realistic in WiFi based HetNets.

On the other hand, the study of wireless transmitters powered by renewable energy has drawn a lot of attention in recent years \cite{SUlukus2015}. Particularly, with noncausal knowledge on energy arrival processes, the throughput maximization problem was investigated for both non-fading and fading channels in \cite{HoZhang2012,OO}, in addition to the classic three-node Gaussian relay channel \cite{HC}. With causal knowledge, the optimal throughput in fading channels over finite-time horizons was obtained via dynamic programming in \cite{HoZhang2012,OO}. A save-then-transmit protocol was proposed in \cite{ruizhang}, where each communication block is divided into two parts: the first one for harvesting energy and the other for data transmission. On the contrary, we consider the save-then-transmit strategy in this paper over an infinite number of communication blocks. For a wireless network where multiple EH-based users share one common channel, authors in \cite{MACprotocols} investigated the performance of some standard medium access control protocols, e.g., TDMA, framed-Aloha, and dynamic-framed-Aloha. Under the similar system setup, authors in \cite{NMichelusi2015} proposed a decentralized access scheme based on game theory, which could achieve some local maxima of the network utility. In this paper, a different scenario is studied where each user has a multi-channel access, and an individual utility to maximize.

Channel probing techniques have also been studied in the literature. In \cite{Xiaowen}, the authors discussed how a transmitter probes a relay channel with some additional time cost when its direct channel is undesirable. In addition, similar channel probing and selection problems for WiFi and cognitive radios were investigated in \cite{VK} and \cite{Tshu}, respectively. For \cite{Xiaowen,VK,Tshu}, the key idea is that the sender may spend time on probing the channel quality before starting a transmission. We here adopt a similar idea. However, we need to face a different and more challenging scenario: Besides probing the large cell network, we also need to probe the resource availability in the small local network, as well as the local battery status that is dynamic due to the energy arrival and withdrawal.

The remainder of this paper is organized as follows. The specific system model and problem formulation are described in Section \ref{sysandprob}. The throughput optimization problem is solved for both Markovian and i.i.d. cases in Section \ref{solution}. The optimal power allocation with traditional power supplies is discussed in Section \ref{special}. Numerical results are provided in Section \ref{numerical}. Finally, Section \ref{fin} concludes the paper.

\section{System Model and Problem Formulation}\label{sysandprob}
\subsection{System Model}\label{systemmodel}
As shown in Fig. \ref{chamodel}, an uplink HetNet communication scenario is considered: One private channel connected to the large network BS is assigned to each EH-based MU, and one common channel connected to a given small network AP is randomly accessed by all nearby users. All private and common channels are orthogonal in frequency, slotted in time, and synchronized. The duration of each time slot is unified. Moreover, in each slot, an MU can access at most one local AP through the common channel. Define the probability that the common channel is secured by an MU as $p_s$, called {\it securing probability}. Similar to a WiFi system, the MU cannot hold the common channel forever; it is required to release the common channel after the usage.

\subsubsection{Channel model}
Under the above setup, an MU can fulfill a transmission: i) via the private channel only; ii) or via both the private and common channels.
\begin{itemize}
  \item In case i), the received signal in the $t$-th time slot at the BS is given by
\begin{equation}\label{BSsig}
   y_t=h_t\sqrt{P_t}x_t+z_t,
\end{equation}
where $h_t$ is the channel gain of the MU-to-BS link, $P_t$ is the transmit power, $x_t$ is the transmitted signal with zero mean and unit variance, and $z_t$ is the circularly symmetric complex Gaussian (CSCG) noise with zero mean and unit variance. Define $\{H_t=|h_t|^2\}$ on a state space $\mathcal{H}$ with finite mean and variance.
  \item In case ii), the received signal in the $t$-th time slot at the BS is the same as (\ref{BSsig}), and that at the AP is given by
\begin{align}
   &y^{c}_t=h^{c}_t\sqrt{P_t^c}x^{c}_t+z^{c}_t,
\end{align}
where $h^{c}_t$ is the channel gain of the MU-to-AP link, $P_t^c$ is the transmit power over the common channel, $x^{c}_t$ and $z^{c}_t$ are defined similarly as in (\ref{BSsig}). Define $\{H_t^c=|h^{c}_t|^2\}$ on a space $\mathcal{H}_{c}$ with finite mean and variance.
\end{itemize}
Here, we assume that $H_t$ follows a more general Markovian model \cite{QZhang1999} while $H^c_t$ follows an i.i.d. model, due to the fact that the MU-to-BS link usually experiences a much longer distance such that the channel may be under correlated shadowing, while the MU-to-AP link usually experiences fast fading, given its much shorter distance. The CSI includes both $H_t$ and $H_t^{c}$. For simplicity, the time for the MU to learn the CSI is neglected given the much longer length of one time slot.

Assume that the fiber connections between the BS/AP and the C-RAN are perfect, such that the C-RAN based joint decoding is optimal. By applying the Shannon capacity formula, at time slot $t$, the instant transmission rate $R_t$ of the MU over the above channel model is expressed as
\begin{equation}
  R_t=\left\{
        \begin{array}{ll}
          \log\left(1+H_tP_t\right), & \hbox{via the private channel only;} \\
          \log\left(1+H_tP_t\right)+\log\left(1+H_t^{c} P_t^{c}\right), & \hbox{via both the private and common channels.}
        \end{array}
      \right.\nonumber
\end{equation}
Note that the common channel can be secured with probability $p_s$ in our proposed multi-channel model. To make the expression of $R_t$ more concise, we introduce an indicator $\phi_t$ such that
\begin{equation}
  \phi_t=\left\{
           \begin{array}{ll}
             1, & \hbox{with probability (w.p.) $p_s$;} \\
             0, & \hbox{w.p. $1-p_s$.}
           \end{array}
         \right.\nonumber
\end{equation}
Then, $R_t$ can be written as
\begin{align}\label{rate}
R_t=\log\left(1+H_tP_t\right)+\phi_t\log\left(1+ H_t^{c} P_t^{c}\right).
\end{align}
The constraint on transmit power levels $P_t$ and $P_t^c$ will be specified later\footnote{Note that even when $\phi_t=1$, $P_t^c$ may still be assigned as zero by our protocol, as explained later.}.

\subsubsection{Energy model}
In general, the entire operation of the MU relies on the harvested energy. Here, we mainly focus on the effect of the EH constraint on transmit power, not only for analytical tractability and gaining insights, but also due to the fact that data transmission usually dominates the power consumption in medium-to-long range wireless systems \cite{Cui2004,VB2006}. In other words, the energy consumption on circuit overhead and channel training (acquiring CSI of both private and common channels) are assumed relatively negligible.

We use $\left\{B_t\right\}_{t\geq1}$ to denote the energy level at the battery for the considered MU at the beginning of time slot $t$, and quantify the energy level into unit steps, i.e., $B_t\in\mathcal{B}=\{0, \delta, 2\delta,\ldots, B_{max}\delta\}$, where $\delta$ is the smallest energy unit, and $B_{max}$ could be either a finite integer or infinity. For the case of $B_{max}=+\infty$, it is a good approximation when the battery capacity is large enough compared with the EH rate, e.g., an AA-sized NiMH battery has a capacity of 7.7 kJ, which requires a couple of hours to be fully charged by some commercial solar panels \cite{VSharma}. During time slot $t$, the MU harvests $E_t$ amount of energy, where the sequence $\left\{E_{t}\right\}_{t\geq1}$ is modeled as a homogeneous Markov process \cite{NM}. Due to hardware limitations, the EH rate could be represented over a finite state space $\mathcal{E}\subseteq\left\{E:E=k\delta, k\in\mathds{N}\bigcup\{0\}\right\}$. The energy state information (ESI, i.e., EH rate and battery status) is assumed causally known by the MU.


\subsubsection{Operation model}
Given that the MU is driven by the accumulated energy, we consider a ``save-then-transmit'' scheme over multiple time slots: The MU harvests energy and exploits the access opportunity of the common channel simultaneously over a certain number of time slots, and then transmits by using up the total available energy in the battery. Such a scheme has the nature of maximizing the short-term transmission rate, and is practical due to its implementation simplicity. As such, if we let $t=1$ as the first time slot after one data transmission, $B_{t}$ can be written as
\begin{equation}
  B_{t}=\min\left\{\sum_{i=1}^{t-1}E_i,B_{max}\delta\right\}.\nonumber
\end{equation}
When $t=1$, there is $B_{1}=E_0$, where $E_0$ is the accumulated energy during the transmission slot in the previous save-then-transmit period. The MU decides when to stop ``saving'' and start a transmission according to its current CSI and ESI. Specifically, at the beginning of time slot $t$, according to some optimal save-then-transmit policy, an MU can:
\begin{itemize}
  \item either transmit immediately during the current time slot (via either the private channel or both the private and common channels);
  \item or skip transmission (release the common channel if it has been secured by the MU).
\end{itemize}
\begin{figure}
  \centering
  \includegraphics[width=3in]{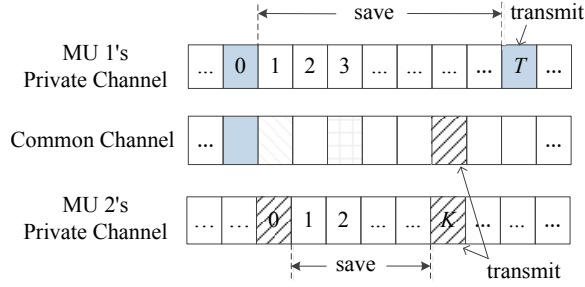}
  \caption{A realization of the proposed save-then-transmit scheme in multi-channel access.}
  \label{sysmodel}
\end{figure}

In Fig. \ref{sysmodel}, we show one realization of the saving and access process, in which two users are assigned with two private channels, respectively, and share one common channel. In particular, MU 1 transmits only through its private channel at time $T$ and MU 2 transmits via both its private and the common channel at time $K$.

\subsection{Problem Formulation}\label{formulation}
Our goal is to maximize the average throughput of the MU. First, we determine the transmit power for maximizing the instant rate $R_t$. At time $t$, according to the save-then-transmit scheme, it is easy to see that the transmit power $P_t$ and $P_t^{c}$ satisfy $P_t+\phi_tP_t^{c}=B_t$. When $\phi_t=0$, it follows $P_t=B_t$, since the MU can only use the private channel; and when if $\phi_t=1$, in order to maximize $R_t$, the power allocation follows the ``water-filling'' scheme given in the next lemma.
\begin{Lemma}\label{powallo}
When the MU can transmit via both the private and common channels (i.e., $\phi_t=1$), it is optimal to allocate power as follows:
\begin{itemize}
  \item If $\left|\frac{1}{H_t^{c}}-\frac{1}{H_t}\right|< B_t$, we have that $P_t=\frac{1}{2}\left(B_t+\frac{1}{H_t^c}-\frac{1}{H_t}\right)$ and $P_t^c=\frac{1}{2}\left(B_t+\frac{1}{H_t}-\frac{1}{H_t^c}\right)$;
  \item If $\left|\frac{1}{H_t^{c}}-\frac{1}{H_t}\right|\geq B_t$ and $H_t> H_t^{c}$, we have $P_t=B_t$ and $P_t^c=0$;
  \item If $\left|\frac{1}{H_t^{c}}-\frac{1}{H_t}\right|\geq B_t$ and $H_t< H_t^{c}$, we have $P_t=0$ and $P_t^c=B_t$.
\end{itemize}
\end{Lemma}

Lemma \ref{powallo} can be proved by using standard convex optimization techniques and thus the proof is omitted for brevity. For notation simplicity, we define the state of the MU, including CSI and ESI, at time $t$ as $\mathbf{F}_t=\{\phi_t,B_t,E_{t-1},H_t,H_t^{c}\}\in\mathcal{F}=\{0,1\}\times\mathcal{B}\times\mathcal{E}\times\mathcal{H}\times\mathcal{H}_c$. In this way, $R_t=R(\mathbf{F}_t)$ is fully determined by $\mathbf{F}_t$.

Next, we let $T$ be some stopping rule indicating the time slot to stop saving and start transmission. Thus, the transmission rate at the time slot $T$ would be denoted as $R(\mathbf{F}_T)$. Here, we make the following assumption: The steady-state distribution of $\{B_t\}$ exists under the stopping rule $T$. We will verify this assumption later by showing that our proposed transmission scheme will indeed result in a stationary $\{B_t\}$. With the above assumption, it follows that the steady-state distribution of $\{\mathbf{F}_t\}$ also exists given that $\{E_t\}$ and $\{H_t\}$ are stationary, respectively. Then, applying the stopping rule $T$ for infinitely many times, we obtain
\begin{equation}
\frac{\underset{L\rightarrow\infty}{\lim}\frac{1}{L}\sum_{l=1}^LR(\mathbf{F}_{T_l})}{\underset{L\rightarrow\infty}{\lim}\frac{1}{L}\sum_{l=1}^LT_l} =\frac{\mathbb{E}[R(\mathbf{F}_T)]}{\mathbb{E}[T]}=\lambda,\nonumber
\end{equation}
where the expectation is taken over the stationary distribution of $\mathbf{F}_t$ and $T$, and $\lambda$ is the average throughput per save-then-transmit period. The core of the proposed save-then-transmit scheme is to find the optimal stopping rule $T^*$ to achieve the maximum throughput $\lambda^*$, which are defined as
\begin{equation}\label{optithrou}
   \lambda^*\triangleq\sup_{T\geq1}\frac{\mathbb{E}[R(\mathbf{F}_T)]}{\mathbb{E}[T]},~T^*\triangleq\arg\sup_{T\geq1}\frac{\mathbb{E}[R(\mathbf{F}_T)]}{\mathbb{E}[T]}.
\end{equation}
In the next section, we will find $T^*$ and $\lambda^*$.


\section{Optimal Stopping Rule and Throughput}\label{solution}
The problem defined in (\ref{optithrou}) is a ``rate-of-return'' problem and could be converted into a standard optimal stopping problem \cite{somepro,optstop}. With some $\lambda>0$ and, we let $G_T(\lambda)=R(\mathbf{F}_T)-\lambda T$, and consider a new problem:
\begin{equation}\label{newproblem}
   \sup_{T\geq1}\mathbb{E}[G_{T}(\lambda)].
\end{equation}
Under this interpretation, $R(\mathbf{F}_T)$ can be regarded as the offer at time $T$, $\lambda T$ is the cost, and $G_T(\lambda)$ is the net reward. We let $G_{\infty}=-\infty$ since it is irrational that a transmitter does not send any data forever. The following lemma, which is directly from Theorem 1 of chapter 6 in \cite{optstop}, connects problems (\ref{optithrou}) and (\ref{newproblem}):
\begin{Lemma}\label{lemmaequi}
$i)$ If (\ref{optithrou}) holds, it follows that when $\lambda=\lambda^*>0$, $\sup_{T\geq1}\mathbb{E}[G_{T}(\lambda^*)]=0$ and the supreme is attained at the same $T^*$ in (\ref{optithrou}); and $ii)$ conversely, if for some $\lambda^*>0$, $\sup_{T\geq1}\mathbb{E}[G_{T}(\lambda^*)]=0$ and it is attained by some $T^*$, then (\ref{optithrou}) holds.
\end{Lemma}

Therefore, we just need to focus on finding the optimal stopping rule $T^*$ for problem (\ref{newproblem}) and $\lambda=\lambda^*>0$ such that $\sup_{T\geq1}\mathbb{E}[G_{T}(\lambda^*)]=0$. In the rest of this section, we first solve problem (\ref{newproblem}) for the case with Markovian private channel states and EH rates. Then, we consider the corresponding i.i.d. case.

\subsection{Markovian Case}\label{sollutiongeneral}
Here, we assume that $\{H_t\}_{t\ge1}$ and $\{E_t\}_{t\ge1}$ are homogeneous Markov processes with some stationary distributions, respectively. Given some $\lambda>0$, we define the remaining expected maximum reward starting at time $t$ in state $\mathbf{F}_t$ as
\begin{align}
  V_t(\mathbf{F}_t)&=\sup_{T\ge t}\mathbb{E}\left[R(\mathbf{F}_T)-\lambda T\mid\mathbf{F}_t\right].\label{remainopt}
\end{align}
Moreover, we observe that the ``cost'' $\lambda$ is a constant, which allows us to use $V_1(\mathbf{F}_t)$ to represent $V_t(\mathbf{F}_t)$, i.e.,
\begin{align}
  V_t(\mathbf{F}_t)&=\sup_{T\ge t}\mathbb{E}\left[R(\mathbf{F}_T)-\lambda(T-(t-1))\mid\mathbf{F}_t\right]-\lambda(t-1)\nonumber\\
  &=\sup_{T\ge 1}\mathbb{E}\left[R(\mathbf{F}_T)-\lambda T)\mid\mathbf{F}_t\right]-\lambda(t-1)\nonumber\\
  &=V_1(\mathbf{F}_t)-\lambda(t-1).\label{vtv1}
\end{align}
Based on this observation, the following proposition shows that the optimal stopping rule for problem (\ref{newproblem}) exists and also shows the form of the optimal stopping rule, whose proof is given in Appendix A.

\begin{Proposition}\label{stopruleexist}
The optimal stopping rule $T^*$ for problem (\ref{newproblem}) exists with either $B_{max}<+\infty$ or $B_{max}=+\infty$, and it has the following form
\begin{align}\label{optirulev1}
  T^*=\min\left\{t\geq1:R(\mathbf{F}_t)-\lambda^*=V_1(\mathbf{F}_t)\right\}.
\end{align}
Moreover, the optimal throughput $\lambda^*$ satisfies
\begin{equation}\label{optithroux}
    \lambda^*=\mathbb{E}\left[\max\left\{R(\mathbf{F}_1),\mathbb{E}\left[V_1(\mathbf{F}_{2})\mid\mathbf{F}_1\right]\right\}\right],
\end{equation}
where $\mathbf{F}_1$ is the initial state of each save-then-transmit period, a random vector defined over the space $\mathcal{F}_1\subseteq\mathcal{F}$ with a certain stationary distribution.
\end{Proposition}

It is observed from (\ref{optirulev1}) that the optimal stopping rule for problem (\ref{newproblem}) is state-dependent and has a threshold-based structure with parameter $\lambda^*$. The structure is derived from the optimality equation (see Theorem 2 in \cite{somepro}), or equivalently, the dynamic programming equation (see (3) in \cite{Huiwang}). Such a structure also implies that the closed form of calculation of $\lambda^*$ is in general extremely difficult, especially in Proposition \ref{stopruleexist} where the stationary distribution of the battery is unknown and the battery capacity could be infinite. Thus, numerical methods are more preferred in finding $\lambda^*$.

Although the calculation of $\lambda^*$ is hard, some properties of $\lambda^*$ can be obtained and are given in the next proposition.
\begin{Proposition}\label{monoV}
$\lambda^*$ is uniquely determined by (\ref{optithroux}) and is strictly increasing over $p_s$.
\end{Proposition}
\begin{proof}
We first show the uniqueness of $\lambda^*$. We observe that in (\ref{optithroux}), its left-hand side is monotonically increasing from zero to positive infinity over $\lambda^*\in[0,+\infty)$. Notice that in the right-hand side of (\ref{optithroux}), we have
\begin{align}
  \mathbb{E}\left[V_1(\mathbf{F}_{2})\mid\mathbf{F}_1\right]
  =\mathbb{E}\left[\left.\sup_{T\ge 1}\mathbb{E}\left[R(\mathbf{F}_T)-\lambda^* T\mid\mathbf{F}_2\right]\right|\mathbf{F}_1\right],\nonumber
\end{align}
which is obtained according to (\ref{remainopt}). It follows that the right-hand side of (\ref{optithroux}) is monotonically deceasing from a finite number, i.e., from
\begin{equation}
  \mathbb{E}\left[\max\left\{R(\mathbf{F}_1),\mathbb{E}\left[\left.\sup_{T\ge 1}\mathbb{E}\left[R(\mathbf{F}_T)\mid\mathbf{F}_2\right]\right|\mathbf{F}_1\right]\right\}\right],\nonumber
\end{equation}
to negative infinity over $\lambda^*\in[0,+\infty)$. Thus, there exists a unique $\lambda^*$ that makes (\ref{optithroux}) hold.

For the monotonicity of $\lambda^*$ over $p_s$, please see Appendix B.
\end{proof}
\begin{Remark}
The strict monotonicity of the optimal throughput $\lambda^*$ over the securing probability $p_s$ implies that the common channel is helpful in general.
\end{Remark}
\begin{Remark}
The stationary distribution of $\{B_t\}$ exists under the optimal stopping rule $T^*$ in (\ref{optirulev1}). Specifically:
\begin{itemize}
  \item When $B_{max}$ is finite, the transition probability of the energy level is also determined under the stopping rule $T^*$ and the stationary distribution of $E_t$. Moreover, all the attainable states of the battery form a positive recurrent class. Thus, $\{B_t\}$ has a steady-state distribution.
  \item When $B_{max}$ is infinite, from the perspective of queueing theory, the average discharging rate of the battery is the same as the recharging rate since all energy will be used for transmission in each save-then-transmit period. Therefore, the stationary distribution of $\{B_t\}$ exists. Moreover, it can be approximated as a Brownian motion process \cite{HChen2001}.
\end{itemize}
\end{Remark}

\subsection{i.i.d. Case}\label{infibattery}
In this subsection, we focus on the case when $\{H_t\}_{t\geq1}$ and $\{E_t\}_{t\geq1}$ are both i.i.d., respectively. As a special case of the one studied in the previous subsection, the optimal stopping rule of this case still exists. Taking one step further, the corresponding optimal stopping rule is simplified to bear a pure-threshold structure.
\begin{Proposition}\label{stoppingrule2}
When $\{H_t\}_{t\geq1}$ and $\{E_t\}_{t\geq1}$ are i.i.d. with finite means and variances, respectively, the optimal stopping rule $T^*$ for problem (\ref{newproblem}) has the following form:
\begin{align}\label{optistoprule2}
     T^*=\min\left\{t\geq1:R(\mathbf{F}_t)\geq\gamma\right\},
\end{align}
where $\gamma$ is a fixed real number.
\end{Proposition}
\begin{proof}
  Since the optimal stopping rule is given by (\ref{optirulev1}) based on Proposition \ref{stopruleexist}, we could further rearrange the rule as
\begin{align}
     T^*&=\inf\left\{t\geq1:V_1(\mathbf{F}_t)-R(\mathbf{F}_t)+\lambda^*=0\right\}\nonumber\\
     &=\inf\left\{t\geq1:\Lambda(\mathbf{F}_t)=0\right\}.\nonumber
\end{align}
The function $\Lambda(\cdot)$ is defined by $\Lambda(\mathbf{F}_t)=V_1(\mathbf{F}_t)-R(\mathbf{F}_t)+\lambda^*$, where $\mathbf{F}_t=\{\phi_t,B_t,E_{t-1},H_t,H^c_t\}\in\mathcal{F}$. The following properties of $\Lambda(\mathbf{F}_t)$ play a key role in the proof of this proposition:
\begin{enumerate}
  \item $\Lambda(\mathbf{F}_t)\geq0$ for all $\mathbf{F}_t$;
  \item $\mathbb{E}[\Lambda(\mathbf{F}_{t})\mid B_{t}]<+\infty$ for all $B_t\ge 0$. Moreover, $\mathbb{E}[\Lambda(\mathbf{F}_{t})\mid B_{t}]\rightarrow0$ as $B_t\rightarrow\infty$;
  \item $\mathbb{E}[\Lambda(\mathbf{F}_{t+1})\mid\mathbf{F}_{t}]<+\infty$  for all $B_t\ge 0$. Moreover, $\mathbb{E}[\Lambda(\mathbf{F}_{t+1})\mid\mathbf{F}_{t}]\rightarrow0$ as $B_t\rightarrow\infty$;
  \item $\Lambda(\mathbf{F}_t)\rightarrow0$ as $R(\mathbf{F}_t)\rightarrow\infty$.
\end{enumerate}
If all the above properties are true, it follows that $\forall\epsilon>0$, there exists $\gamma\geq0$ such that $\Lambda(\mathbf{F}_t)\leq\epsilon$ whenever $R(\mathbf{F}_t)>\gamma$, which implies that the stopping rule $T^*$ has the form given by (\ref{optistoprule2}) (similar to the technique used in \cite{somepro}). The proof of the four properties is given in Appendix C.
\end{proof}

Moreover, we note that the expected value of the optimal stopping rule $T^*$ indicates the mean saving time. The next proposition shows that for a fixed threshold, the mean saving time is shortened under the proposed opportunistic scheme with multi-channel access.
\begin{Proposition}\label{twithps}
  Given a fixed $\gamma>0$, $\mathbb{E}\left[T^*\right]$ is decreasing polynomially over $p_s$.
\end{Proposition}
The proof is given in Appendix D. Following Proposition \ref{stoppingrule2} and Lemma \ref{lemmaequi}, we have
\begin{align}
  0&=\sup_{T\in \mathcal{T}_1}\mathbb{E}[R(\mathbf{F}_{T})-\lambda^*T]\nonumber\\
   &=\mathbb{E}\left[R\left(\mathbf{F}_{T^*}\right)1_{\left\{R(\mathbf{F}_{T^*})\geq\gamma\right\}}\right]-\lambda^*\mathbb{E}[T^*].\nonumber
\end{align}
Then, we obtain
\begin{equation}\label{throughputspe}
  \lambda^*=\max_{\gamma\geq0}\frac{\mathbb{E}\left[R\left(\mathbf{F}_{T^*}\right)1_{\left\{R(\mathbf{F}_{T^*})\geq\gamma\right\}}\right]}{\mathbb{E}\left[T^*\right]}.
\end{equation}

{\bf Conjecture}: $\lambda^*$ is a quasi-concave function over $\gamma$.

Our conjecture will be validated via numerical results in Section \ref{numerical}. Such a conjecture enables us to apply some simple search methods, e.g., bisection search, to find the optimal threshold.

\section{Throughput with Conventional Power Supply}\label{special}
In this section, we investigate the throughput of the MU with a conventional power supply in the discussed multi-channel access system, which will serve as performance benchmarks for our proposed schemes introduced in previous sections. Note that we only need to change the EH constraints into the average power constraints in the setup, and keep the same channel and access models as before.

With a conventional power supply, the instant transmission rate $R_t$ given by (\ref{rate}) still holds. Then, finding the optimal power allocation is equivalent to solving the following optimization problem:
\begin{align}
&\max_{\{P_t,P^c_t\}}~~~\lim_{K\rightarrow\infty}\frac{1}{K}\sum_{t=1}^K\left(\log\left(1+H_tP_t\right)+\phi_t\log\left(1+ H^c_tP^c_t\right)\right)\label{opicapa}\\
&~\hbox{s.t.}~~~~~\lim_{K\rightarrow\infty}\frac{1}{K}\sum_{t=1}^K\left(P_t+\phi_t P^{c}_t\right)\leq \overline{P};\label{averB}\\
&~~~~~~~~~P_t,P^{c}_t\geq0,~\hbox{for all}~t\ge1,\nonumber
\end{align}
where $\overline{P}$ is the maximum average total power. The optimal power allocation is given in the next proposition.
\begin{Proposition}\label{optipowerallo}
The optimal power allocation of problem (\ref{opicapa}) is given as
\begin{equation}
P^*_t=\left(\frac{1}{\xi^*}-\frac{1}{H_t}\right)^+,
P^{c,*}_t=\left\{
           \begin{array}{ll}
             \left(\frac{1}{\xi^*}-\frac{1}{H^c_t}\right)^+, & \hbox{if $\phi_t=1$,} \\
             0, & \hbox{if $\phi_t=0$,}
           \end{array}
         \right.\nonumber
\end{equation}
where $\xi^*$ satisfies the average power constraint (\ref{averB}).
\end{Proposition}
Proposition \ref{optipowerallo} can be proved by applying a similar technique to the proof of optimal adaptation (5) in \cite{Goldsmith}, and thus is omitted here.

The optimal power allocation has a ``water-filling'' structure similar to the optimal solution of the single fading channel case under an average power constraint, while the water level is jointly determined by the securing probability $p_s$ and the statistics of both private and common channels.

\section{Numerical Results}\label{numerical}
In this section, we present some numerical results to validate our analysis. Besides the optimal power allocation with a conventional power supply, we also consider the method of best-effort delivery \cite{David} as a comparison benchmark, i.e., the transmitter directly uses up the harvested energy in the previous time slot and does not store energy. In the simulation, the length of each time slot is $1$ ms and the energy step is set to be $\delta=10^{-3}$ J.

\subsection{Markovian Private Channel Gains}
\begin{figure}
  \centering
  \includegraphics[width=3.7in]{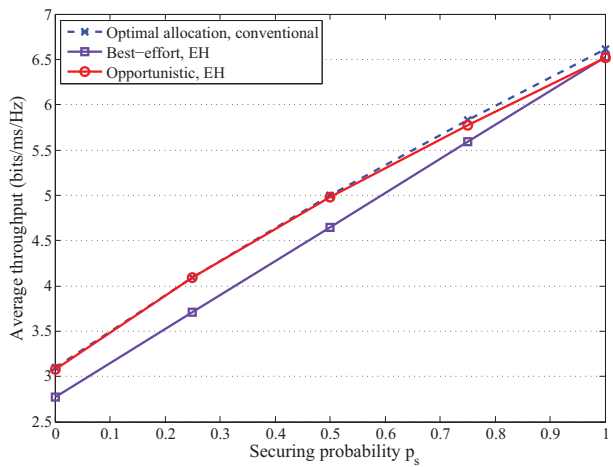}
  \caption{Average throughput vs. securing probability $p_s$ over Markovian channel.}
  \label{finitebatteryps}
\end{figure}

First, we consider a renewable energy supply at the MU with a time-correlated private channel, which corresponds to Section \ref{sollutiongeneral}. Here, we use a simple model to illustrate the throughput performance with different schemes. Let the capacity of the battery $B_{max}=1$, and the EH rate $E_t=\delta$. The common channel is static with a constant power gain $H_t^c=2^5$ for $t\geq1$. The gain of the private channel has two states $\left\{0.1,2^{4}\right\}$ with transition probability 1 from state $0.1$ to $2^{4}$ and probability $0.5$ from state $2^{4}$ to $0.1$.

In Fig.~\ref{finitebatteryps}, we show the average throughput with the opportunistic scheme proposed in Section \ref{sollutiongeneral} against other schemes under the impact of securing probability $p_s$. First, we observe that the optimal allocation with the conventional power supply serves as the performance upper bound. We also observe that the throughput attained by the opportunistic scheme increases as $p_s$ increases, which agrees with Proposition \ref{monoV}. Second, when $0\leq p_s<1$, the opportunistic scheme is better than the best-effort delivery. It agrees with our intuition that when the transmitter experiences a bad channel ($H_t=0.1$), it skips the transmission immediately and waits for a better channel state, which may lead to a higher average throughput. Third, the opportunistic scheme and the best-effort delivery have the same performance at $p_s=1$, since the common channel is good ($H_t^c=2^5$) and always secured by the MU, such that the MU does not need to skip any transmission, which results in the same average throughput for the two schemes. We could also conclude that only when the difference between the good and bad channel conditions is large enough, the opportunistic scheme performs significantly better.

\subsection{i.i.d. Private Channel Gains}
We apply a two-state EH model (similar to that in \cite{NM}), where the EH rate can be either zero (``BAD'') or $4\delta$ (``GOOD'') with probability $0.5$ for each state. The channel gains in either the private or the common channel are i.i.d. following an exponential distribution with unit mean.

\begin{figure}
  \centering
  \includegraphics[width=3.7in]{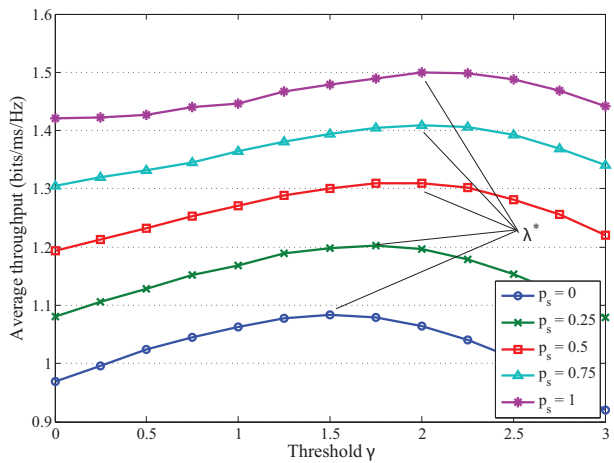}
  \caption{Average throughput vs. threshold $\gamma$.}
  \label{lamgammps}
\end{figure}
In Fig.~\ref{lamgammps}, we show how the threshold $\gamma$ influences the average throughput with different securing probability $p_s$. We observe that the average throughput could be optimized by adjusting the threshold $\gamma$, which validates the results in (\ref{throughputspe}) and our conjecture in Section \ref{infibattery}.

\begin{figure}
  \centering
  \includegraphics[width=3.7in]{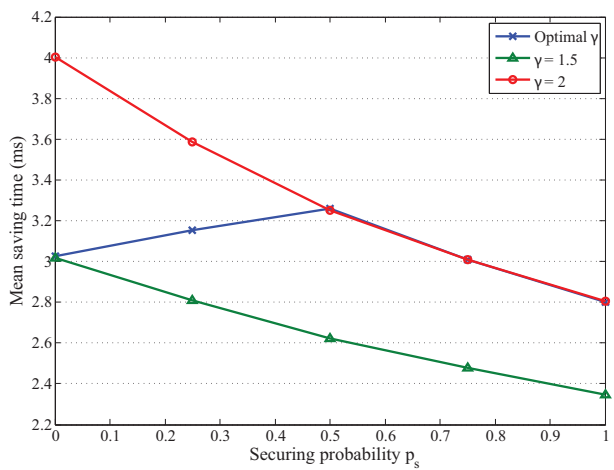}
  \caption{Mean saving time vs. securing probability $p_s$.}
  \label{delayps}
\end{figure}
We also show how the mean saving time varies over the securing probability $p_s$ in Fig. \ref{delayps}. Since the optimal threshold is different when $p_s$ changes, we choose two typical values for comparison: $\gamma=1.5$, which is optimal for $p_s=0$; and $\gamma=2$, which is optimal for $p_s=0.5,~0.75,~1$ based on our results in Fig.~\ref{lamgammps}. For either $\gamma=1.5$ or $\gamma=2$, we observe from Fig. \ref{delayps} that the mean saving time decreases as $p_s$ increases, which agrees with Proposition \ref{twithps}. The mean saving time with optimal $\gamma$ falls in between those with $\gamma=1.5$ and $\gamma=2$, respectively.

\begin{figure}
  \centering
  \includegraphics[width=3in]{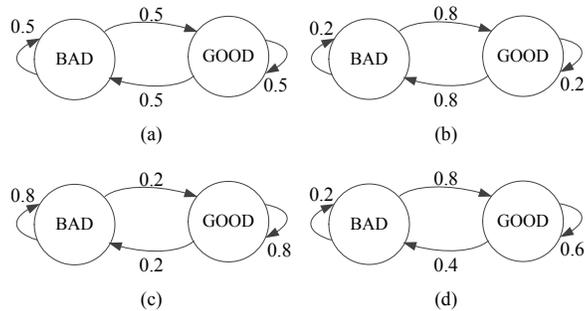}
  \caption{Different EH models.}
  \label{ehmodels}
\end{figure}

\begin{figure}
  \centering
  \includegraphics[width=3.7in]{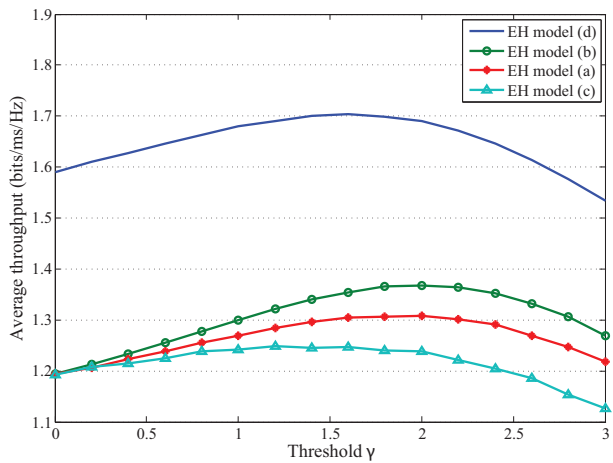}
  \caption{Average throughput vs. threshold $\gamma$ with different EH models.}
  \label{lamgammenergy}
\end{figure}
Next, we want to show the existence of EH diversity and the proposed opportunistic scheme is able to explore this type of diversity. For better illustration, we focus on the pure-threshold policy discussed in Section \ref{infibattery}. The securing probability $p_s$ is set to be 0.5. The Markovian EH model (a) in Fig. \ref{ehmodels} is the benchmark, which is equivalent to an i.i.d. EH model with probability 0.5 to be either ``GOOD'' or ``BAD''. To compare, we choose EH models (b) and (c) as shown in Fig. \ref{ehmodels}, which have the same stationary distribution as that of model (a), while bearing different ``randomness'': Model (b) changes from one state to the other with a higher frequency compared with model (a), and model (c) changes with a lower frequency such that the EH rate is likely to stay in one state and rarely change over time. In addition, we also consider model (d), which represents the case that the EH rate has a higher stationary probability to be ``GOOD''.

In Fig.~\ref{lamgammenergy}, we show the average throughput over different threshold values for the four EH models with the pure-threshold opportunistic scheme depicted in Fig. \ref{ehmodels}. First, we observe that EH model (d) achieves the highest throughput, since model (d) has the largest stationary probability for the EH rate to be in the ``GOOD'' state. Second, we observe the throughput differences across EH models (a), (b) and (c). When $\gamma=0$, these three models lead to the same throughput performance, for $\gamma=0$ implies that the opportunistic transmission scheme is not applied such that the average throughput is mainly determined by the stationary characteristics. When $\gamma$ increases until the optimal value that leads to the maximum average throughput, we observe that the EH model (b) could make the transmitter achieve a slightly higher throughput than model (a). Similarly, model (a) is able to achieve a higher throughput than that of model (c). Note that among models (a), (b) and (c), EH model (b) is more likely to shift from one state to the other, while model (c) is likely to keep staying in either ``BAD'' or ``GOOD'' state. The EH model (a) behaves in between. The observation is that when the EH rate varies in a more dramatic way, it has larger randomness, where we could claim a higher EH diversity. Accordingly, our proposed opportunistic scheme would take advantage of such EH diversity by exploiting the EH variation, where the transmitter could opportunistically wait or start the transmission depending on the energy state.

\begin{figure}
  \centering
  \includegraphics[width=3.7in]{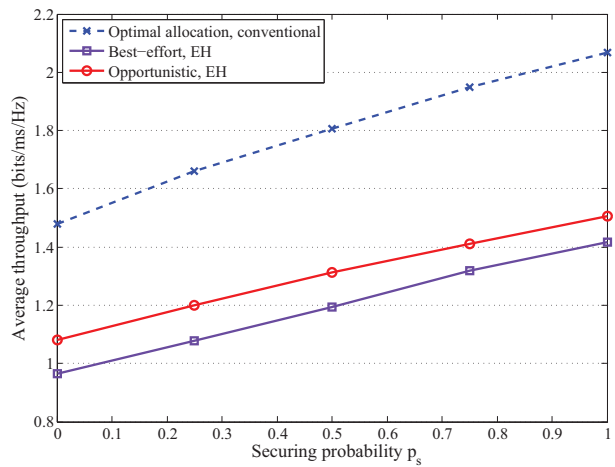}
  \caption{Average throughput vs. securing probability $p_s$ over i.i.d. channel.}
  \label{throughputps}
\end{figure}
Finally, over i.i.d. channel, the throughput performance of the MU with different power supplies and transmission schemes is shown in Fig. \ref{throughputps}. To make them comparable, we let $\overline{P}=2$ W. The EH-based transmitter with the opportunistic transmission scheme could achieve about $70\%$ of the throughput with the optimal power allocation, which is relatively worse than the Markovian case as shown in Fig \ref{finitebatteryps}.

\section{Conclusion}\label{fin}
In this paper, we considered a HetNet uplink with multi-channel access, where each EH-powered MU has deterministic access to a private channel linked to the cellular BS, and random access to a common channel linked to a local AP. As such, the MU could fulfil a transmission via its private channel or via both private and common channels. By jointly taking advantage of channel-energy variation and common channel sharing, we proposed an opportunistic transmission scheme that allows the transmitter to properly probe the channel-energy state, such that the average transmission rate is maximized. In particular, we formulated the average throughput maximization problem as an optimal stopping problem of rate-of-return. By applying the optimal stopping theory, we proved that the optimal stopping rule exists and has a state-dependent and threshold-based structure in general. Moreover, when the private channel gains and EH rates are i.i.d., respectively, the optimal stopping rule turned out to be a simple pure-threshold policy. We also found the optimal power allocation scheme for the transmitter powered by a conventional power supply, to serve as performance benchmarks. Numerical results validated the analysis with both Markovian and i.i.d. statistical models for the private channel gains and EH rates. We showed that under a renewable energy supply, the proposed opportunistic transmission scheme could achieve a higher throughput than the method of best-effort delivery. Also, our simulation results revealed the throughput gap between the cases with  conventional and renewable energy supplies. Furthermore, the phenomenon of EH diversity was briefly discussed, which could be explored by the proposed pure-threshold policy such that the throughput performance could be enhanced.

\section*{Appendices}
\subsection{Proof of Proposition \ref{stopruleexist}}\label{stoprulefinite}
According to the optimal stopping theory \cite{somepro, optstop}, the existence of the optimal stopping rule could be proved by checking the following two conditions: For a given $\lambda>0$,
\begin{description}
  \item[C1:] $\mathbb{E}\left[\sup_{T\geq1}G_T(\lambda)\right]<\infty$;
  \item[C2:] $\limsup_{T\rightarrow\infty}G_{T}(\lambda)\le G_{\infty}=-\infty$ a.s..
\end{description}
We first check C1 and C2 for $B_{max}<+\infty$ and $B_{max}=+\infty$, respectively.
\begin{itemize}
  \item $B_{max}<+\infty$: For C1, we have $\sup_{T\geq1}G_T(\lambda)\leq\sup_{T\geq1}R(\mathbf{F}_T)$. Since the channel gains are finite a.s., and the battery capacity is finite, the expectation of the transmission rate is finite as well, which proves that C1 holds. For C2, we only need to show that for any large negative real number $\nu<0$, there exists $K\geq0$ a.s. such that for all $T\geq K$, $G_T(\lambda)=R(\mathbf{F}_T)-\lambda T<\nu$. In fact, for any $T$, $\mathbb{E}\left[R(\mathbf{F}_T)\right]<\infty$, which implies that $\mathbb{P}\left\{R(\mathbf{F}_T)=\infty\right\}=0$. However, the term $\lambda T$ will increase to infinity as $T\rightarrow\infty$. Thus, when $T\geq K$, $R(\mathbf{F}_T)-\lambda T$ can be as small as we want a.s., i.e., $R(\mathbf{F}_T)-\lambda T<\nu$ a.s., which proves that C2 holds.

  \item $B_{max}=+\infty$: For this case, we check C2 first. Recall the expression of $R(\mathbf{F}_T)$ in (\ref{rate}) and $B_T$ given as $B_{T}=\sum_{i=1}^{T-1}E_i\le E_{max}T$, where $E_{max}$ is the maximum EH rate and is finite. Then, we have
\begin{align}
&R(\mathbf{F}_T)-\lambda T\nonumber\\
\leq &\log\left(\frac{1+H_TB_T}{2^{\lambda T/2}}\right)+\log\left(\frac{1+H_T^{c} B_T}{2^{\lambda T/2}}\right)\nonumber\\
\leq &\log\left(\frac{1+H_TE_{max}T}{2^{\lambda T/2}}\right)+\log\left(\frac{1+H_T^{c} E_{max}T}{2^{\lambda T/2}}\right),\label{lastterm}
\end{align}
a.s.. By using L'H$\hat{o}$pital's rule \cite{Chatterjee2005}, the first term in (\ref{lastterm}) satisfies
\begin{align}
   \lim_{T\rightarrow\infty}\frac{1+H_TE_{max}T}{2^{\lambda T/2}}&\le\lim_{T\rightarrow\infty}\frac{H_TE_{max}}{\frac{\lambda\ln2}{2}2^{\lambda T/2}}=0, ~\hbox{a.s..}\nonumber
\end{align}
We could apply a similar check for the second term of (\ref{lastterm}). Thus, C2 holds. For C1, we could use the above results of C2 and obtain that $\forall\epsilon>0$, there exists an $N>0$ such that $\mathbb{E}\left[\sup_{T\geq1}G_T(\lambda)\right]<\mathbb{E}\left[\sup_{1\leq T\leq N}(R(\mathbf{F}_T)-\lambda T)\right]+\epsilon$. Since the channel gains are finite a.s., and for all $1\leq T\leq N$, $\mathbb{E}[B_T]=\mathbb{E}\left[\sum_{i=1}^{T-1}E_i\right]<\infty$, we obtain $\mathbb{E}\left[\sup_{1\leq T\leq N}(R(\mathbf{F}_T)-\lambda T)\right]<\infty$, which implies that C1 holds.
\end{itemize}
Therefore, both C1 and C2 hold for either $B_{max}<+\infty$ or $B_{max}=+\infty$, which implies that the optimal stopping rule exists.

Next, we derive the optimal stopping rule. According to (\ref{vtv1}), we have $V_t(\mathbf{F}_t)=V_1(\mathbf{F}_t)-\lambda(t-1)$. Meanwhile, $V_t(\mathbf{F}_t)$ satisfies the dynamic programming equation (equation (3) in \cite{Huiwang}):
\begin{align}
  V_t(\mathbf{F}_t)=&\max\left\{R(\mathbf{F}_t)-\lambda t,\mathbb{E}\left[V_{t+1}(\mathbf{F}_{t+1})\mid\mathbf{F}_t\right]\right\}.\label{dyna}
\end{align}
Therefore, the optimal stopping rule has the following form
\begin{align}
  T^*=&\min\left\{t\geq1:R(\mathbf{F}_t)-\lambda t=V_t(\mathbf{F}_t)\right\}\nonumber\\
=&\min\left\{t\geq1:R(\mathbf{F}_t)-\lambda t=V_1(\mathbf{F}_t)-\lambda(t-1)\right\}\nonumber\\
=&\min\left\{t\geq1:R(\mathbf{F}_t)-\lambda=V_1(\mathbf{F}_t)\right\},\nonumber
\end{align}
where the second equation holds due to (\ref{vtv1}). By letting $\lambda=\lambda^*$, we obtain the form of $T^*$ as shown in (\ref{optirulev1}).

Finally, we compute $\lambda^*$. By Lemma \ref{lemmaequi}, $\lambda^*$ makes the following equation hold:
\begin{align}
    0=&\sup_{T\ge1}\mathbb{E}[G_{T}(\lambda)]\nonumber\\
=&\mathbb{E}\left[\max\left\{R(\mathbf{F}_1 )-\lambda^*,\mathbb{E}\left[V_2(\mathbf{F}_2)\mid\mathbf{F}_1 \right]\right\}\right]\nonumber\\
=&\mathbb{E}\left[\max\left\{R(\mathbf{F}_1 )-\lambda^*,-\lambda^*+\mathbb{E}\left[V_1(\mathbf{F}_2)\mid\mathbf{F}_1 \right]\right\}\right].\nonumber
\end{align}
Thus, we could obtain $\lambda^*$ by some simple rearrangements.

\subsection{Proof of Proposition \ref{monoV}}
Recall Proposition \ref{stopruleexist} that the optimal stopping rule $T^*$ exists and it is easy to check that $T^*\in\left\{T\geq 1:\mathbb{E}\left[T\right]<\infty\right\}$. Then, given some $\epsilon>0$, there exists an $M\ge2$ such that for all $t\ge M$, we have $\mathbb{P}(T^*=t)<\epsilon$. Therefore, when we consider the expected value of $V_{1}(\mathbf{F}_{1})$, we can just focus on a finite horizon, i.e., $1\le t\le M$. Then, by the dynamic programming algorithm (e.g., Theorem 2 of Chapter 3 in\cite{optstop}, or equation (3) in \cite{Huiwang}), we have
\begin{align}
  &V_{1}(\mathbf{F}_{t})=\max\left\{R(\mathbf{F}_t),\mathbb{E}\left[V_{1}(\mathbf{F}_{t+1})\mid\mathbf{F}_t\right]\right\}-\lambda^*,\nonumber\\
  &~~~~~~~~~~~\hbox{for $t=1,2,\ldots,M-1$}\nonumber\\
  &V_{1}(\mathbf{F}_{M})=R(\mathbf{F}_M)-\lambda^*.\nonumber
\end{align}

Now, we show that $\lambda^*$ is strictly increasing over $p_s$ by contradiction. First, we fix $\lambda^*$, and let $p_s$ increase to $p_s+\Delta$, where $\Delta$ is a small positive real number. Then, we move backward. Note that at step $t=M$, $V_{1}(\mathbf{F}_{M})$ only depends on $\mathbf{F}_{M}$ and does not change with $p_s$. At $t=M-1$, we observe that
\begin{align}
  &\mathbb{E}\left[V_{1}(\mathbf{F}_{M})\mid\mathbf{F}_{M-1}\right]\nonumber\\
   =&(p_s+\Delta)\mathbb{E}\left[R(H_M,H_M^c)-R(H_M,0)\mid\mathbf{F}_{M-1}\right]\nonumber\\
   &+\mathbb{E}\left[R(H_M,0)\mid\mathbf{F}_{M-1}\right]-\lambda^*.\nonumber
\end{align}
Note that the private channel could not be strictly better than the common channel \cite{JGAndrews2012}, i.e., it is unrealistic that $\min_{H_M\in\mathcal{H}}H_M>\max_{H_M^c\in\mathcal{H}_c}H_M^c$. It follows that $\mathbb{E}\left[R(H_M,H_M^c)-R(H_M,0)\mid\mathbf{F}_{M-1}\right]>0$. Thus, we have that $\mathbb{E}\left[R(\mathbf{F}_M)-\lambda^*\mid\mathbf{F}_{M-1}\right]$ strictly increases as $p_s$ increases to $p_s+\Delta$.

Suppose that at $t=k$ for $2\le k\le M-1$, $\mathbb{E}\left[V_{1}(\mathbf{F}_{k+1})\mid\mathbf{F}_{k}\right]$ strictly increases as $p_s$ increases to $p_s+\Delta$. Since the expected value of $R(\mathbf{F}_{k})$ also strictly increases following a similar argument as we discussed at step $t=M-1$, we have that the expected value of $\max\left\{R(\mathbf{F}_{k}),\mathbb{E}\left[V_{1}(\mathbf{F}_{k+1})\mid\mathbf{F}_{k}\right]\right\}$ strictly increases. Then, at $t=k-1$, we have
\begin{align}
  &\mathbb{E}\left[V_{1}(\mathbf{F}_{k})\mid\mathbf{F}_{k-1}\right]\nonumber\\
  =&\mathbb{E}\left[\max\left\{R(\mathbf{F}_{k}),\mathbb{E}\left[V_{1}(\mathbf{F}_{k+1})\mid\mathbf{F}_{k}\right]\right\}\mid\mathbf{F}_{k-1}\right]-\lambda^*,
\end{align}
which strictly increases and thus implies that such an increment holds for all $t=1,2,\ldots,M-1$.

At the step $t=1$, we have
\begin{align}
  \mathbb{E}[V_{1}(\mathbf{F}_{1})]= \mathbb{E}\left[\max\left\{R(\mathbf{F}_{1}),\mathbb{E}\left[V_{1}(\mathbf{F}_{2})\mid\mathbf{F}_{1}\right]\right\}\right]-\lambda^*,
\end{align}
where $\mathbb{E}\left[\max\left\{R(\mathbf{F}_{1}),\mathbb{E}\left[V_{1}(\mathbf{F}_{2})\mid\mathbf{F}_{1}\right]\right\}\right]$ should also strictly increase as $p_s$ increases to $p_s+\Delta$. However, we recall from Proposition \ref{stopruleexist} that $\mathbb{E}[V_{1}(\mathbf{F}_{1})]=0$, which is attained by $T^*$ and $\lambda^*$. It implies that in order to make $\mathbb{E}[V_{1}(\mathbf{F}_{1})]=0$, the value $\lambda^*$ should not be fixed and must strictly increase accordingly, which contradicts the assumption in the first step that $\lambda^*$ is fixed. Thus, $\lambda^*$ strictly increases as $p_s$ increases. Finally, this proposition is proved by letting $\epsilon\rightarrow0$ (i.e., $M$ is large enough).

\subsection{Proof of Proposition \ref{stoppingrule2}}

For Property 1), it is straightforward to see that
\begin{align}
     \Lambda(\mathbf{F}_t)=&V_1(\mathbf{F}_t)-R(\mathbf{F}_t)+\lambda^*\nonumber\\
=&\max\left\{R(\mathbf{F}_t)-\lambda^*,-\lambda^*+\mathbb{E}\left[V_1(\mathbf{F}_{t+1})\mid \mathbf{F}_t\right]\right\}\nonumber\\
     &-R(\mathbf{F}_t)+\lambda^*\nonumber\\
=&\max\left\{0,\mathbb{E}\left[V_1(\mathbf{F}_{t+1})\mid\mathbf{F}_t\right]-R(\mathbf{F}_t)\right\}\geq0.\label{phiv2}
\end{align}

For Property 2), suppose that the transmitter does not stop channel-energy probing until time $t$; then starting at $t$, we should have $T\in\left\{T\geq t:\mathbb{E}\left[T\right]<\infty\right\}$. Thus, $\mathbb{E}[\Lambda(\mathbf{F}_t)\mid B_t]$ could be written as
\begin{align}
\mathbb{E}[\Lambda(\mathbf{F}_t)\mid B_t]=\sum_{n\ge t}\mathbb{P}(T=n)\mathbb{E}[\Lambda(\mathbf{F}_t)\mid B_t,T=n]<\infty,\nonumber
\end{align}
due to $\mathbb{P}(T=+\infty)=0$. Then, with a fixed $T=n$ such that $t\le n<\infty$, along with Property 1), $\mathbb{E}[\Lambda(\mathbf{F}_t)\mid B_t,n]$ is expanded as
\begin{align}
0\leq &\mathbb{E}[\Lambda(\mathbf{F}_t)\mid B_t,n]\nonumber\\
=&\mathbb{E}\left[R(\mathbf{F}_n)-R(\mathbf{F}_t)-\lambda^*n\mid B_t\right]+\lambda^*\nonumber\\
\leq&(1-p_s)\mathbb{E}\left[\log\left(\frac{1+HB_n}{1+HB_t}\right)\right]\label{term1}\\
+&p_s\left(\mathbb{E}\left[\log\left(\frac{1+HP_n}{1+HP_t}\right)+\log\left(\frac{1+H^{c}P_n^{c}}{1+H^{c}P_t^{c}}\right)\right]\right),\label{term2}
\end{align}
where the second inequality holds due to $-\lambda^*n+\lambda^*\leq0$ for $n\geq t$. Note that we do not put the time index $n$ on $H$ and $H^c$ since $\{H_t\}_{t\geq1}$ and $\{H_t^c\}_{t\geq1}$ are i.i.d., respectively. Next, we want to show that both (\ref{term1}) and (\ref{term2}) are finite and could be as small as we want with a large $B_t$, which would complete the proof for 2).
\begin{itemize}
  \item For (\ref{term1}): by plugging $B_n=B_t+\sum_{i=t}^{n-1}E_i$, we obtain
\begin{align}
(\ref{term1})=(1-p_s)\mathbb{E}\left[\log\left(1+\frac{H\sum_{i=t}^{n-1}E_i}{1+HB_t}\right)\right]<+\infty\nonumber
\end{align}
since $H$ has finite mean and $\{E_j\}_{t\leq j\leq n-1}$ are i.i.d. with finite mean as well. Moreover, if $B_t\rightarrow\infty$, $(\ref{term1})\rightarrow0$.
  \item For (\ref{term2}): Since $P_n+P_n^c=B_n$, and both $H$ and $H^c$ have finite means, respectively, it follows that (\ref{term2}) is finite. When the transmitter occupies the common channel at time $T\geq t$, there are three possible events by Lemma \ref{powallo}: If $\left|\frac{1}{H^{c}}-\frac{1}{H}\right|\geq B_n$, allocating all power to one of the two channels; otherwise, allocating the power to both channels at a certain ratio. Note that the probability of any above events happening does not depend on $n$ if $B_t$ is large enough. To see this point, we let
      \begin{align}
        &Q=\mathbb{P}\left(\left|\frac{1}{H^{c}}-\frac{1}{H}\right|<B_t\right),\nonumber\\
        &q_1=\mathbb{P}\left(\left|\frac{1}{H^{c}}-\frac{1}{H}\right|\geq B_t,H>H^c\right),\nonumber\\
        &q_2=\mathbb{P}\left(\left|\frac{1}{H^{c}}-\frac{1}{H}\right|\geq B_t,H<H^c\right).\nonumber
      \end{align}
      When $B_t$ is large, there is
\begin{align}
\mathbb{P}\left(\left|\frac{1}{H^{c}}-\frac{1}{H}\right|< B_t+\sum_{i=t}^{n-1}E_i\right)\approx Q,\nonumber
\end{align}
and similarly, we have
\begin{align}
  &\mathbb{P}\left(\left|\frac{1}{H^{c}}-\frac{1}{H}\right|\geq B_n,H>H^c\right)\approx q_1,\nonumber\\
  &\mathbb{P}\left(\left|\frac{1}{H^{c}}-\frac{1}{H}\right|\geq B_n,H<H^c\right)\approx q_2.\nonumber
\end{align}
Then, by applying $Q$, $q_1$ and $q_2$, we can expand (\ref{term2}) as
\begin{align}
  &(\ref{term2})\approx\nonumber\\
&p_s\left(q_1\mathbb{E}\left[\log\left(\frac{1+HB_n}{1+hB_t}\right)\right]+q_2\mathbb{E}\left[\log\left(\frac{1+H^cB_n}{1+h^cB_t}\right)\right]\right)\nonumber\\
  &~~+p_sQ\mathbb{E}\left[\log\frac{\left(1+HB_n+\frac{H}{H^{c}}\right)\left(1+H^{c}B_n+\frac{H^{c}}{H}\right)}{\left(1+HB_t+\frac{H}{H^{c}}\right)\left(1+H^{c}B_t+\frac{H^{c}}{H}\right)}\right].\nonumber
\end{align}
Similarly as the reasoning in (\ref{term1}), we obtain that $(\ref{term2})\rightarrow0$ as $B_t\rightarrow\infty$.
\end{itemize}
Therefore, we conclude that $\mathbb{E}[\Lambda(\mathbf{F}_t)\mid B_t]$ is finite and could be arbitrarily small when $B_t$ is large enough.

For 3), we expend $\mathbb{E}[\Lambda(\mathbf{F}_{t+1})\mid\mathbf{F}_{t}]$ as
\begin{align}
&\mathbb{E}[\Lambda(\mathbf{F}_{t+1})\mid\mathbf{F}_{t}]=\mathbb{E}[\Lambda(\mathbf{F}_{t+1})\mid B_{t}]\nonumber\\
=&\sum_{e\in\mathcal{E}}\mathbb{P}(E_{t}=e)\mathbb{E}[\Lambda(\mathbf{F}_{t+1})\mid B_{t}],\nonumber
\end{align}
since only $\{B_{t}\}$ are correlated over time. By Property 2), we know $\mathbb{E}[\Lambda(\mathbf{F}_{t+1})\mid B_{t}]$ is finite and thus $\mathbb{E}[\Lambda(\mathbf{F}_{t+1})\mid\mathbf{F}_{t}]$ is finite since $\mathcal{E}$ is a finite space. Moreover, by Property 2), we have $\mathbb{E}[\Lambda(\mathbf{F}_{t+1})\mid B_t]\rightarrow0$ as $B_t\rightarrow\infty$. Therefore, it follows that $\mathbb{E}[\Lambda(\mathbf{F}_{t+1})\mid\mathbf{F}_{t}]$ could be as small as we want when $B_t$ is large enough.

By now, we are ready to show Property 4). We could rewrite (\ref{phiv2}) as
\begin{align}
     \Lambda(\mathbf{F}_t)&=\max\left\{0,\mathbb{E}\left[V_1(\mathbf{F}_{t+1})\mid\mathbf{F}_t\right]-R(\mathbf{F}_t)\right\}\nonumber\\
&=\max\left\{0,\mathbb{E}\left[\Lambda(\mathbf{F}_{t+1})+R(\mathbf{F}_{t+1})-\lambda^*\mid\mathbf{F}_t\right]-R(\mathbf{F}_t)\right\}.\nonumber
\end{align}
Next, we show Property 4) by contradiction. Suppose that $\Lambda(\mathbf{F}_t)>0$ for all $R(\mathbf{F}_t)\geq0$, we have
\begin{align}
     \mathbb{E}\left[\Lambda(\mathbf{F}_{t+1})+R(\mathbf{F}_{t+1})\mid\mathbf{F}_t\right]>R(\mathbf{F}_t)+\lambda^*.\label{compare}
\end{align}
For the left-hand side of (\ref{compare}), $\mathbb{E}\left[R(\mathbf{F}_{t+1})\mid\mathbf{F}_t\right]$ is finite for any fixed $B_t$, and $\mathbb{E}\left[\Lambda(\mathbf{F}_{t+1})\mid\mathbf{F}_t\right]$ is either a finite number or a arbitrarily small positive number if $B_t$ is large enough. Then, we choose $K<+\infty$ and $B_t=B_{max}$ such that the left-hand side of  (\ref{compare}) is upper-bounded by $K$. With such $K$ and $B_t$, we have
\begin{align}
     K&>\mathbb{E}\left[\Lambda(\mathbf{F}_{t+1})+R(\mathbf{F}_{t+1})\mid\mathbf{F}_t\right]>R(\mathbf{F}_t)+\lambda^*.\label{comparerate}
\end{align}
However, for the right-hand side of (\ref{compare}) with the same $B_t$, $R(\mathbf{F}_t)$ could be arbitrarily large if $H_t$ and $H_t^c$ are large enough. Then, there always exists an $M>0$ such that when $H_t,H_t^c>M$, $R(\mathbf{F}_t)>K$, which leads to the contradiction with the inequality (\ref{comparerate}). Therefore, we obtain that $\Lambda(\mathbf{F}_t)=0$ when $R(\mathbf{F}_t)$ is large enough.

Overall, we have shown that all four properties hold, and we conclude that the optimal stopping rule has a pure-threshold structure given by (\ref{optistoprule2}).

\subsection{Proof of Proposition \ref{twithps}}
Given some $\gamma>0$, we let $q_t(p_s)=\mathbb{P}\left(R(H_t,H_t^c)\geq\gamma\right)$. Based on the form of the stopping rule $T^*$ given by (\ref{optistoprule2}), we obtain
\begin{align}
\mathbb{E}\left[T^*\right]=q_1(p_s)+\sum_{t=2}^{\infty}tq_t(p_s)\prod_{n=1}^{t-1}(1-q_n(p_s)).\nonumber
\end{align}
Since $\mathbb{E}\left[T^*\right]<\infty$, it follows that $\forall\epsilon,\epsilon_0>0$, there exists $N>0$ such that $\mathbb{P}\left(T^*=t\right)=q_t(p_s)\prod_{n=1}^{t-1}(1-q_n(p_s))<\epsilon$ for all $t\ge N$, and $\sum_{t=N}^{\infty}tq_t(p_s)\prod_{n=1}^{t-1}(1-q_n(p_s))<\epsilon_0$. Note that the generality still holds by letting $q_N(p_s)=\epsilon$ since $\mathbb{P}(T^*=N)=\epsilon\prod_{n=1}^{N-1}(1-q_n(p_s))<\epsilon$. Then, we have
\begin{align}
&\mathbb{E}\left[T^*\right]=q_1(p_s)+\sum_{t=2}^{N}tq_t(p_s)\prod_{n=1}^{t-1}(1-q_n(p_s))+\epsilon_0\nonumber\\
=&\epsilon_0+q_1(p_s)+(1-q_1(p_s))\cdot\left(~2q_2(p_s)+(1-q_2(p_s))\cdot\right.\nonumber\\
&~~\left.\cdots\left(~(N-1)q_{N-1}(p_s)+(1-q_{N-1}(p_s))N\epsilon~\right)\cdots\right).\nonumber
\end{align}
We introduce $U_t=tq_t(p_s)+(1-q_t(p_s))U_{t+1}=t+(1-q_t(p_s))\left(U_{t+1}-t\right)$, where we notice $U_{t+1}-t>0$. With this notation, we have $\mathbb{E}\left[T^*\right]=\epsilon_0+U_{1}$.

Next, we show the monotonicity of $\mathbb{E}\left[T^*\right]$ by using the mathematical induction in a ``backward'' fashion: From a very large number $N$ back to $t=1$. First, we check $U_N$. It is true since $U_{N}=N\epsilon$, which is independent of $p_s$. Then, suppose that $U_{k+1}$ is decreasing over $p_s$ for $k=2,\ldots,N-1$; we check $U_k=k+(1-q_k(p_s))(U_{k+1}-k)$. For $q_k(p_s)$, we have
\begin{align}
q_k(p_s)=&\mathbb{P}\left(R(H_k,0)\geq\gamma\right)+\nonumber\\
&p_s\left(\mathbb{P}\left(R(H_k, H_k^c)\geq\gamma\right)-\mathbb{P}\left(R(H_k,0)\geq\gamma\right)\right),\nonumber
\end{align}
where $\mathbb{P}\left(R(H_k, H_k^c)\geq\gamma\right)\geq\mathbb{P}\left(R(H_k,0)\geq\gamma\right)$ due to $R(H_k, H_k^c)\geq R(H_k,0)$. It follows that $q_k(p_s)$ is an increasing linear function of $p_s$, and then $1-q_k(p_s)$ is deceasing. Since both $(1-q_k(p_s))$ and $(U_{k+1}-k)$ are nonnegative and decreasing, $U_k$ is decreasing as well. Moreover, $U_k$ is a polynomial function of $p_s$ due to the linearity of $q_k(p_s)$ and the iteration function, i.e., $U_k=k+(1-q_k(p_s))(U_{k+1}-k)$. Thus, we obtain that $\mathbb{E}[T^*]=U_{1}+\epsilon_0$ is a polynomial function and decreasing over $p_s$. By letting $\epsilon_0\rightarrow0$, we are done with the proof for this proposition.

\end{spacing}
\end{document}